\documentclass[11pt]{article}

\usepackage[utf8]{inputenc}
\usepackage[T1]{fontenc}
\usepackage{amsmath,amssymb,amsthm}
\usepackage{algorithm}
\usepackage{algpseudocode}
\usepackage{graphicx}
\usepackage{tikz}
\usetikzlibrary{calc, fit, backgrounds, positioning}
\usepackage{hyperref}
\usepackage{cleveref}
\usepackage{booktabs}
\usepackage{xcolor}
\usepackage[margin=1in]{geometry}
\usepackage{authblk}

\newtheorem{theorem}{Theorem}[section]

\newtheorem{corollary}[theorem]{Corollary}
\newtheorem{proposition}[theorem]{Proposition}
\newtheorem{definition}[theorem]{Definition}
\newtheorem{remark}[theorem]{Remark}

\title{Flashback: A Reversible Bilateral Run-Peeling Decomposition of Strings}
\author[1,2]{Thomas Konstantinovsky}
\author[3]{Gur Yaari}

\affil[1]{Faculty of Engineering, Bar-Ilan University, Ramat Gan 5290002, Israel}
\affil[2]{Bar-Ilan Institute for Nanotechnology and Advanced Materials, Bar-Ilan University, Ramat Gan, Israel}
\affil[3]{Department of Pathology, Yale School of Medicine, New Haven, CT 06510, United States}
\date{\today}

\begin{document}

\maketitle

\begin{abstract}
We introduce \textsc{Flashback}, a reversible string decomposition that
repeatedly peels the maximal leading and trailing character runs from a
sentinel-wrapped input, recording each pair as one \emph{bilateral token}.
Decomposition and reconstruction both run in $O(n)$ time and space.
Our central result is a \emph{run-pairing theorem}: Flashback is
equivalent to pairing the first run of the string with the last, the second
with the second-to-last, and so on.  This gives an exact token count of
$1+\lceil r/2\rceil$ for a string with $r$ maximal runs, and matches a
lower bound that holds for any admissible bilateral run-peeling scheme.
From the run-pairing theorem the main structural properties follow as
corollaries: the irreducible peeling kernel uses at most two symbols;
palindromes are precisely the strings whose run-length encoding is
symmetric with an odd number of runs; the image of the decomposition
admits an explicit finite-state characterisation; and changing one run
length rewrites exactly one content token.
\end{abstract}

\section{Introduction}
\label{sec:introduction}

Most string decompositions process the input from one end to the other.
Lempel--Ziv~\cite{lz76,lz77,lz78} parses left-to-right against a dictionary
of previously seen substrings; the Burrows--Wheeler
Transform~\cite{bwt} permutes characters to cluster identical symbols;
run-length encoding walks the string once and collapses adjacent duplicates.
Each of these captures something real (repetition, clustering, local
homogeneity), but none exposes the \emph{outside-in} structure of a string:
the way a string's boundary runs can be peeled away, layer by layer, to
reveal an inner core.

We study one such decomposition.  \textsc{Flashback} wraps the input in
sentinels and, at each step, removes the maximal run from each end at
once, recording the two peeled runs as a single \emph{bilateral token}
together with the split between them.  The companion procedure
$\mathcal{F}^{-1}$ reconstructs the string by nesting the tokens from the
inside out.  The algorithm is deterministic, runs in $O(n)$ time and space,
and is exactly reversible.

The algorithm is short; the structural theory it opens up is the point
of this paper.  Once the decomposition is phrased at the run level, a
single observation drives most of what follows: Flashback is equivalent
to pairing the first run of the string with the last, the second with
the second-to-last, and so on.  From that one identification the token
count, the irreducible kernel, a palindrome test, a description of the
image of $\mathcal{F}$, and an exact edit-locality theorem all fall out
as corollaries.  Flashback is not a compression scheme (the total
character count across tokens equals $n+2$); it is a structural tool for
reasoning about run-level symmetry, and to our knowledge the
decomposition has not been studied before.

\paragraph{Contributions.}
We contribute three groups of results:

\begin{enumerate}
    \item A definition and analysis of Flashback and its inverse, with
          proofs of exact reversibility, uniqueness, and linear time
          and space (Sections \ref{sec:algorithm}--\ref{sec:complexity}).
    \item A \emph{run-pairing theorem} identifying Flashback with
          bilateral pairing of the run-length encoding, giving an exact
          token count of $1+\lceil r/2\rceil$.  A matching lower bound
          shows that no admissible bilateral run-peeling scheme produces
          fewer tokens (\Cref{thm:run_pairing,thm:optimality}).
    \item Structural consequences of the run-pairing theorem: the
          peeling kernel has alphabet size at most two; palindromes are
          exactly the strings whose RLE is symmetric with an odd number
          of runs; the image of $\mathcal{F}$ admits an explicit
          characterisation; and changing one run length rewrites
          exactly one content token.
\end{enumerate}

\section{Preliminaries}
\label{sec:preliminaries}

We work with finite strings over an alphabet $\Sigma$.  For $s\in\Sigma^*$
with $|s|=n$, we write $s[i..j]$ for the inclusive 1-indexed substring and
$s[i..j)$ for its half-open counterpart (empty when $i>j$).

\begin{definition}[Maximal runs]
\label{def:runs}
The \emph{leading run} of a non-empty string $w = w_1\cdots w_m$ is its
longest prefix of a single repeated symbol; the \emph{trailing run} is its
longest suffix of a single repeated symbol.  Both are uniquely determined
and have length $\ge 1$.
\end{definition}

\begin{definition}[Sentinel-wrapped string]
\label{def:sentinel}
Fix two symbols $\texttt{@},\texttt{\$}\notin\Sigma$.  The
\emph{sentinel-wrapped form} of $s$ is $\hat{s} = \texttt{@}\cdot s\cdot
\texttt{\$}$, of length $n+2$.  The sentinels guarantee that at depth~$0$
the leading and trailing runs are distinct length-one runs.
\end{definition}

\begin{definition}[Bilateral token]
\label{def:token}
A \emph{bilateral token} is a pair $\tau=(\sigma,p)$ with $\sigma\in\Sigma^+$
(the concatenation of a front run and a back run) and $p\in\mathbb{N}_0$
the \emph{split position}: $\sigma[1..p]$ is the front, $\sigma[p{+}1..|\sigma|]$
the back.  A token is \emph{terminal} when $p=0$.  We write $\tau$ as
\texttt{$\sigma$\_$p$}; e.g.\ $(\texttt{CFF},1)$ is \texttt{CFF\_1}.
\end{definition}

\begin{definition}[Decomposition operators]
\label{def:operators}
$\mathcal{F}(s)$ denotes the token sequence produced by
\textsc{Flashback-Decompose}$(s)$, and $\mathcal{F}^{-1}(T)$ the string
recovered by \textsc{Flashback-Reverse}$(T)$.
\end{definition}

\section{The Flashback Algorithm}
\label{sec:algorithm}

Flashback consists of two procedures: a \emph{decomposition}
$\mathcal{F}$ that turns a string into a token sequence, and a
\emph{reconstruction} $\mathcal{F}^{-1}$ that inverts it exactly.

\subsection{Decomposition}
\label{sec:decompose}

The idea is simple.  We wrap the input string in two sentinel characters,
and then repeatedly remove the maximal runs at the two ends, emitting each
pair as a single token.  The sentinels $\texttt{@}$ and $\texttt{\$}$ are
there only to make the algorithm uniform: without them a string that
already begins and ends with the same character would need a special case.

Concretely: on an active span, find the leading run and the trailing run;
concatenate them into a symbol string and record the split position
(the length of the leading run); then recurse on the middle.  The
recursion stops when the middle is empty: either the span itself is
empty, or the front and back runs together cover it.

\begin{algorithm}[ht]
\caption{\textsc{Flashback-Decompose}$(s)$}
\label{alg:decompose}
\begin{algorithmic}[1]
\Require String $s \in \Sigma^*$
\Ensure Ordered token sequence $T = [\tau_0, \tau_1, \ldots, \tau_{k-1}]$
\State $\hat{s} \gets \texttt{@} \cdot s \cdot \texttt{\$}$
\State $T \gets [\,]$
\State \Call{Peel}{$\hat{s},\; 1,\; |\hat{s}|{+}1,\; T$}
\State \Return $T$
\end{algorithmic}
\end{algorithm}

\begin{algorithm}[ht]
\caption{\textsc{Peel}$(\hat{s},\; \mathit{lo},\; \mathit{hi},\; T)$}
\label{alg:peel}
\begin{algorithmic}[1]
\Require Sentinel-wrapped string~$\hat{s}$; active window
  $[\mathit{lo},\mathit{hi})$ (half-open); token list~$T$
\If{$\mathit{lo} \ge \mathit{hi}$}
    \State \Return \Comment{empty span}
\EndIf
\Statex
\State \Comment{\textbf{Leading run:}}
\State $\ell \gets 1$
\While{$\mathit{lo} + \ell < \mathit{hi}$ \textbf{and}
       $\hat{s}[\mathit{lo}+\ell] = \hat{s}[\mathit{lo}]$}
    \State $\ell \gets \ell + 1$
\EndWhile
\Statex
\If{$\ell = \mathit{hi} - \mathit{lo}$}
    \State Append $\bigl(\hat{s}[\mathit{lo}..\mathit{hi}),\; 0\bigr)$
           to $T$ \Comment{entire span is one run}
    \State \Return
\EndIf
\Statex
\State \Comment{\textbf{Trailing run:}}
\State $r \gets \mathit{hi} - 1$
\While{$r > \mathit{lo}$ \textbf{and}
       $\hat{s}[r-1] = \hat{s}[\mathit{hi}-1]$}
    \State $r \gets r - 1$
\EndWhile
\Statex
\State $\sigma \gets \hat{s}[\mathit{lo}..\mathit{lo}+\ell) \;\|\;
                     \hat{s}[r..\mathit{hi})$
       \Comment{front $\|$ back}
\State $\mathit{mid\_lo} \gets \mathit{lo} + \ell$;\quad
       $\mathit{mid\_hi} \gets r$
\Statex
\If{$\mathit{mid\_lo} \ge \mathit{mid\_hi}$}
    \State Append $(\sigma,\; 0)$ to $T$
           \Comment{no middle left - terminal token}
\Else
    \State Append $(\sigma,\; \ell)$ to $T$
           \Comment{split position $= \ell$}
    \State \Call{Peel}{$\hat{s},\;\mathit{mid\_lo},\;\mathit{mid\_hi},\; T$}
\EndIf
\end{algorithmic}
\end{algorithm}

A few remarks.  The first token is always $(\texttt{@\$},1)$: because
$\texttt{@}$ and $\texttt{\$}$ are distinct and do not appear in $\Sigma$,
they are always peeled at depth~$0$.  \textsc{Peel} is tail-recursive
(each call makes at most one recursive call, as its last action), so the
procedure converts trivially to a loop (see \Cref{app:iterative_decompose}).
Finally, every character of $\hat{s}$ is \emph{consumed} (placed into a
token's symbol string) exactly once across all levels; the total number of
character inspections remains $O(n)$ as we show in \Cref{sec:complexity}.

\subsection{Reconstruction}
\label{sec:reverse}

Reconstruction is the mirror image.  We allocate an output buffer and walk
the token list from outside in.  For each token we write the front run at
the current position, recurse to fill the gap, and then write the back run
after the inner result has been filled.  Because each character lands
directly in its final slot, the total work is $\Theta(n)$.

\begin{algorithm}[ht]
\caption{\textsc{Flashback-Reverse}$(T)$}
\label{alg:reverse}
\begin{algorithmic}[1]
\Require Token sequence $T = [\tau_0, \tau_1, \ldots, \tau_{k-1}]$
\Ensure Original string $s$
\State Allocate buffer $B$ of size $\sum_{d} |\sigma_d|$
\State $L \gets$ \Call{Nest}{$T,\; 0,\; B,\; 1$}
\State \Return $B[2..L{-}1]$ \Comment{strip sentinels}
\end{algorithmic}
\end{algorithm}

\begin{algorithm}[ht]
\caption{\textsc{Nest}$(T,\; i,\; B,\; \mathit{pos})$}
\label{alg:nest}
\begin{algorithmic}[1]
\Require Token sequence~$T$; index~$i$; buffer~$B$; write position $\mathit{pos}$
\Ensure Number of characters written
\State $(\sigma, p) \gets T[i]$
\If{$i = |T| - 1$} \Comment{innermost token}
    \State Write $\sigma$ to $B[\mathit{pos}..\mathit{pos}+|\sigma|)$
    \State \Return $|\sigma|$
\EndIf
\State Write $\sigma[1..p]$ to $B[\mathit{pos}..\mathit{pos}+p)$
       \Comment{front}
\State $\mathit{inner\_len} \gets$
       \Call{Nest}{$T,\; i{+}1,\; B,\; \mathit{pos}+p$}
\State Write $\sigma[p{+}1..|\sigma|]$ to
       $B[\mathit{pos}+p+\mathit{inner\_len}..)$ \Comment{back}
\State \Return $p + \mathit{inner\_len} + (|\sigma| - p)$
\end{algorithmic}
\end{algorithm}

In words: each token $\tau_i=(\sigma_i,p_i)$ denotes a nested concatenation
$\sigma_i[1..p_i] \;\|\; (\text{inner result}) \;\|\; \sigma_i[p_i{+}1..]$,
with the innermost token returned verbatim.  The algorithm above
realises this in-place by writing each character directly into its final
buffer slot, so the total number of character writes is $n+2$
(\Cref{cor:total_symbol_length}) and no string is ever copied.  A terminal
token ($p=0$) can only appear as the last element of $T$
(\Cref{lem:terminal}).

We prove that decomposition and reconstruction are exact inverses, i.e.\
that $\mathcal{F}^{-1}(\mathcal{F}(s))=s$ for every $s$, in
\Cref{thm:reversibility}.

\subsection{Worked Example}
\label{sec:examples}

Let $s = \texttt{CASSAYFF}$ ($n = 8$), so
$\hat{s} = \texttt{@CASSAYFF\$}$ (length~10).

\begin{center}
\begin{tabular}{ccccl}
\toprule
$d$ & Active span & Front run & Back run & Token \\
\midrule
0 & \texttt{@CASSAYFF\$} & \texttt{@}\,(1) & \texttt{\$}\,(1) & \texttt{@\$\_1} \\
1 & \texttt{CASSAYFF}    & \texttt{C}\,(1) & \texttt{FF}\,(2) & \texttt{CFF\_1} \\
2 & \texttt{ASSAY}       & \texttt{A}\,(1) & \texttt{Y}\,(1)  & \texttt{AY\_1} \\
3 & \texttt{SSA}         & \texttt{SS}\,(2) & \texttt{A}\,(1) & \texttt{SSA\_0} \\
\bottomrule
\end{tabular}
\end{center}

\noindent Reconstruction inside-out:
\texttt{SSA} $\to$
\texttt{A}$\|$\texttt{SSA}$\|$\texttt{Y} = \texttt{ASSAY} $\to$
\texttt{C}$\|$\texttt{ASSAY}$\|$\texttt{FF} = \texttt{CASSAYFF} $\to$
strip sentinels \checkmark.
The RLE of \texttt{CASSAYFF} is
$\texttt{C}^1 \texttt{A}^1 \texttt{S}^2 \texttt{A}^1 \texttt{Y}^1
\texttt{F}^2$ ($r = 6$ runs).  By \Cref{thm:run_pairing},
$k = 1 + \lceil 6/2 \rceil = 4$, with tokens pairing run~1 with run~6
(\texttt{CFF\_1}), run~2 with run~5 (\texttt{AY\_1}), and the middle two
runs forming the peeling kernel \texttt{SSA\_0}.
\Cref{fig:example} shows the peeling visually.

\begin{figure}[ht]
\centering
\begin{tikzpicture}[
  x=0.55cm, y=0.75cm,
  ch/.style={font=\ttfamily, minimum width=0.55cm, minimum height=0.55cm,
             inner sep=0pt, draw=black!60, line width=0.3pt},
  fr/.style={fill=blue!18},
  bk/.style={fill=red!18},
  md/.style={fill=black!5},
  lvl/.style={font=\footnotesize, anchor=east, inner sep=4pt},
  tk/.style={font=\small\ttfamily, anchor=west, inner sep=4pt}
]
\node[lvl] at (-1, 0) {$d=0$};
\node[ch, fr] at (0, 0) {@};
\foreach \c/\i in {C/1, A/2, S/3, S/4, A/5, Y/6, F/7, F/8}
  \node[ch, md] at (\i, 0) {\c};
\node[ch, bk] at (9, 0) {\$};
\node[tk] at (9.5, 0) {@\$\_1};

\node[lvl] at (-1, -1) {$d=1$};
\node[ch, fr] at (1, -1) {C};
\foreach \c/\i in {A/2, S/3, S/4, A/5, Y/6} \node[ch, md] at (\i, -1) {\c};
\node[ch, bk] at (7, -1) {F};
\node[ch, bk] at (8, -1) {F};
\node[tk] at (9.5, -1) {CFF\_1};

\node[lvl] at (-1, -2) {$d=2$};
\node[ch, fr] at (2, -2) {A};
\foreach \c/\i in {S/3, S/4, A/5} \node[ch, md] at (\i, -2) {\c};
\node[ch, bk] at (6, -2) {Y};
\node[tk] at (9.5, -2) {AY\_1};

\node[lvl] at (-1, -3) {$d=3$};
\node[ch, fr] at (3, -3) {S};
\node[ch, fr] at (4, -3) {S};
\node[ch, bk] at (5, -3) {A};
\node[tk] at (9.5, -3) {SSA\_0};
\end{tikzpicture}
\caption{Bilateral peeling of \texttt{CASSAYFF}.  Blue marks the leading
run, red the trailing run, and grey the middle passed to the next depth.
At each level the leading and trailing runs are removed together and
emitted as one token; at $d{=}3$ the front and back runs cover the whole
span, so the level emits a terminal token and the recursion stops.}
\label{fig:example}
\end{figure}

\section{Correctness}
\label{sec:correctness}

This section establishes the two facts that make Flashback a genuine
decomposition: the token sequence is an exact inverse of the string, and
it is the only token sequence consistent with the rules.  We then collect
the small set of invariants that subsequent sections will use repeatedly.

Throughout, let $s\in\Sigma^*$ with $|s|=n$ and
$T = \mathcal{F}(s) = [\tau_0,\ldots,\tau_{k-1}]$ with
$\tau_d = (\sigma_d,p_d)$; let $[\mathit{lo}_d,\mathit{hi}_d)$ be the
active span at depth~$d$.

\begin{remark}[Empty string]
\label{rem:empty}
For $s=\varepsilon$, $\hat{s}=\texttt{@\$}$ is a single span whose leading
and trailing runs cover it entirely, so \textsc{Peel} emits one terminal
token $\tau_0=(\texttt{@\$},0)$.  This is the only degenerate case; all
later statements assume $|s|\ge 1$ unless noted.
\end{remark}

\subsection{Reversibility and Uniqueness}

\begin{theorem}[Reversibility]
\label{thm:reversibility}
$\mathcal{F}^{-1}\bigl(\mathcal{F}(s)\bigr)=s$ for every $s\in\Sigma^*$.
\end{theorem}

\begin{proof}
We prove the stronger local claim that \textsc{Peel} and \textsc{Nest} are
inverses level by level.  Induct on the span length
$m=\mathit{hi}-\mathit{lo}$: the tokens emitted by \textsc{Peel} on
$\hat{s}[\mathit{lo}..\mathit{hi})$ are exactly what \textsc{Nest} needs
to reconstruct that span.

\emph{Base} ($m\le 1$, or the span is a single run).  \textsc{Peel} emits
one terminal token $(\sigma,0)$ with $\sigma=\hat{s}[\mathit{lo}..\mathit{hi})$;
\textsc{Nest} returns $\sigma$ verbatim.

\emph{Inductive step.}  Let the leading run have length $\ell\ge 1$ and the
trailing run start at $r$, with $b=\mathit{hi}-r\ge 1$.  \textsc{Peel} emits
$\tau_d=(\sigma_d,\ell)$ with
$\sigma_d=\hat{s}[\mathit{lo}..\mathit{lo}{+}\ell)\,\|\,\hat{s}[r..\mathit{hi})$
and recurses on the middle $\hat{s}[\mathit{lo}{+}\ell..r)$, whose length
$m-\ell-b<m$ gives the inductive hypothesis.  \textsc{Nest} at depth~$d$
computes
\[
\sigma_d[1..\ell]\;\|\;\textsc{Nest}(T,d{+}1)\;\|\;
\sigma_d[\ell{+}1..|\sigma_d|]
=\hat{s}[\mathit{lo}..\mathit{hi}).
\]
Applying the claim to the full span and stripping sentinels yields $s$.
\end{proof}

\begin{theorem}[Uniqueness]
\label{thm:uniqueness}
$\mathcal{F}(s)$ is the unique sequence of bilateral tokens consistent with
the decomposition rules.
\end{theorem}

\begin{proof}
At each depth, maximality fixes the leading and trailing runs uniquely, so
$\sigma_d$, $p_d$, and the middle boundaries are all forced.  The active span
strictly shrinks at every non-terminal depth (\Cref{prop:nesting}), so the
recursion terminates in finitely many steps, producing a unique finite
sequence.
\end{proof}

\subsection{Structural Invariants}

The following three invariants are used throughout the rest of the paper:
the active spans form a nested chain, the token symbol strings partition
the characters of $\hat{s}$, and the sentinel and terminal tokens occupy
fixed positions in $T$.  All three are straightforward consequences of the
algorithm, but stating them explicitly pays off later.

\begin{proposition}[Nested spans]
\label{prop:nesting}
The active spans form a strictly decreasing chain
$[\mathit{lo}_0,\mathit{hi}_0)\supsetneq[\mathit{lo}_1,\mathit{hi}_1)
\supsetneq\cdots$, with
$\mathit{lo}_{d+1}=\mathit{lo}_d+\ell_d$ and $\mathit{hi}_{d+1}=r_d$.
\end{proposition}

\begin{proof}
At each non-terminal depth, $\ell_d\ge 1$ and $\mathit{hi}_d-r_d\ge 1$, so
the new span is a proper subinterval of the current one.
\end{proof}

\begin{proposition}[Alphabet partition]
\label{thm:alphabet}
The token symbol strings form a multiset partition of $\hat{s}$:
$\biguplus_{d=0}^{k-1}\mathrm{chars}(\sigma_d)=\mathrm{chars}(\hat{s})$.
In particular,
\[
\sum_{d=0}^{k-1}|\sigma_d|=n+2.
\]
\label{cor:total_symbol_length}
\end{proposition}

\begin{proof}
At every depth, $\sigma_d$ contains the characters of the leading and
trailing runs; the middle, passed to the next level, is disjoint from both.
By \Cref{prop:nesting}, induction covers all characters exactly once.
\end{proof}

\begin{proposition}[Sentinel and terminal tokens]
\label{prop:sentinel}
For $|s|\ge 1$ the first token is $\tau_0=(\texttt{@\$},1)$, non-terminal
with split position~$1$.  A terminal token ($p=0$) can appear only as the
last element of $T$.
\label{lem:terminal}
\end{proposition}

\begin{proof}
Since $\texttt{@}\ne\texttt{\$}$ and neither lies in $\Sigma$, both runs at
depth~0 have length~1; the middle is $s$, which is non-empty, so the token
is non-terminal with $\sigma_0=\texttt{@\$}$ and $p_0=1$.  A terminal token
is emitted only when \textsc{Peel} does not recurse, hence no further tokens
follow.
\end{proof}

\section{Complexity}
\label{sec:complexity}

\subsection{Time and Space}

\begin{theorem}[Linear time and space]
\label{thm:time_decompose}
\label{thm:time_reverse}
\label{thm:space}
Both $\mathcal{F}$ and $\mathcal{F}^{-1}$ run in $\Theta(n)$ time and use
$O(n)$ space in the word-RAM model.
\end{theorem}

\begin{proof}
\emph{Decomposition.}  At depth~$d$, scanning the leading and trailing runs
performs $\ell_d+b_d$ successful comparisons and at most two failing ones.
The $\ell_d+b_d$ matched characters are consumed into $\sigma_d$ and never
revisited.  By \Cref{thm:alphabet} their total is $n+2$; failing
comparisons sum to at most $2k\le n+O(1)$.  Token emission is
$O(|\sigma_d|)$, again summing to $n+2$.

\emph{Reconstruction.}  \textsc{Nest} writes each character directly to its
final position, totalling $n+2$ writes (\Cref{cor:total_symbol_length}),
with $O(1)$ bookkeeping per token.

\emph{Space.}  The sentinel-wrapped string, token list, and output buffer
are each $O(n)$.  Decomposition is tail-recursive and admits an iterative
rewrite in $O(1)$ auxiliary space (\Cref{app:iterative_decompose}).
\end{proof}

\subsection{Expected Token Count for Random Strings}
\label{sec:average}

\begin{theorem}[Expected token count]
\label{thm:expected_k}
Let $s$ have length $n\ge 2$ with each character drawn i.i.d.\ uniformly
from an alphabet of size $\sigma\ge 2$, and set $p=(\sigma-1)/\sigma$.  Then
$r(s)-1\sim\mathrm{Bin}(n{-}1,p)$ and
\[
\mathbb{E}[k] \;=\; 1 + \frac{1+(n-1)p}{2}
+ \frac{1+\bigl((2-\sigma)/\sigma\bigr)^{n-1}}{4}
\;=\; \frac{n(\sigma-1)}{2\sigma} + O(1).
\]
\end{theorem}

\begin{proof}
Adjacent characters differ independently with probability $p$, so
$X:=r-1\sim\mathrm{Bin}(n{-}1,p)$.  By \Cref{thm:run_pairing},
$k = 1+\lceil r/2\rceil = 1+\lceil(X+1)/2\rceil$ and
$\lceil(X+1)/2\rceil = (X+1+\mathbf{1}_{X\text{ even}})/2$.  For $X\sim
\mathrm{Bin}(n{-}1,p)$, $\Pr(X\text{ even}) = \tfrac{1}{2}(1+(1-2p)^{n-1})$,
and $1-2p=(2-\sigma)/\sigma$.  Substituting $\mathbb{E}[X]=(n{-}1)p$ gives
the formula; the asymptotic follows from $|(2-\sigma)/\sigma|<1$ for
$\sigma\ge 3$ (and vanishing for $\sigma=2$).
\end{proof}

The full distribution of $k$ is available in closed form since
$k = 1+\lceil(1+X)/2\rceil$; the exact variance and kernel-alphabet
statistics are recorded in \Cref{app:random_string_stats}.

\section{Structural Analysis}
\label{sec:structural}

The algorithm described so far is short, but its real content is
structural: Flashback acts on the run-level shape of a string, not on
individual characters.  Once we make this observation precise in the
\emph{run-pairing theorem} below, most other properties of the
decomposition follow almost mechanically: kernel size, palindromes, edit
locality, and the image of $\mathcal{F}$ are all corollaries of the same
pairing identity.  The rest of the section is organised around that
theorem: first its statement, then a matching lower bound, and finally
the downstream consequences.

\subsection{Flashback as Bilateral Run Pairing}

\begin{definition}[Maximal runs and RLE]
\label{def:rle}
For a non-empty string $s$, the \emph{run-length encoding} (RLE) is the
unique decomposition
\[
  s \;=\; a_1^{m_1}\, a_2^{m_2}\, \cdots\, a_r^{m_r},
\]
where each $a_i \in \Sigma$ is the symbol of the $i$-th run, each
$m_i \ge 1$ is its length, and consecutive symbols differ:
$a_i \ne a_{i+1}$.  The integer $r = r(s)$ is the number of maximal
runs of $s$.  Here $a^m$ abbreviates the string consisting of $a$
repeated $m$ times; for example $\mathtt{CASSAYFF}$ has RLE
$\mathtt{C}^1\, \mathtt{A}^1\, \mathtt{S}^2\, \mathtt{A}^1\, \mathtt{Y}^1\,
\mathtt{F}^2$ and $r = 6$.
\end{definition}

\begin{theorem}[Run-pairing theorem]
\label{thm:run_pairing}
Let $s$ have $r$ maximal runs, and write $h = \lceil r/2 \rceil$.  Then
$\mathcal{F}(s)$ has exactly
\[
  k \;=\; 1 + h
\]
tokens, arranged as follows.
\begin{itemize}
  \item \textbf{Sentinel token.} $\tau_0 = (\texttt{@\$},\; 1)$.
  \item \textbf{Content tokens.} For $1 \le j \le h - 1$, token $\tau_j$
        pairs the $j$-th run from the left with the $j$-th run from the
        right:
        \[
          \tau_j \;=\; \bigl(a_j^{m_j} \cdot a_{r+1-j}^{m_{r+1-j}},\; m_j\bigr).
        \]
  \item \textbf{Terminal token.} $\tau_{k-1}$ records the middle run, or
        the two middle runs if $r$ is even:
        \[
          \tau_{k-1} \;=\;
          \begin{cases}
            \bigl(a_h^{m_h},\; 0\bigr) & \text{if } r \text{ is odd}, \\[3pt]
            \bigl(a_h^{m_h} \cdot a_{h+1}^{m_{h+1}},\; 0\bigr)
                                       & \text{if } r \text{ is even}.
          \end{cases}
        \]
\end{itemize}
In particular, $k \le \lceil n/2 \rceil + 1$, with equality exactly when
every run has length $1$ (i.e.\ $r = n$).
\end{theorem}

\begin{proof}
The sentinels $\texttt{@}$ and $\texttt{\$}$ are distinct and lie outside
$\Sigma$, so at depth $0$ they always peel as $\tau_0 = (\texttt{@\$}, 1)$,
with the middle being $s$ itself.  It therefore suffices to count the
tokens \textsc{Peel} emits on a bare (unsentineled) span with $r$ runs,
and show this equals $h$ with the claimed shape.

We proceed by induction on $r$.  Since $a_i \ne a_{i+1}$ in any RLE, the
leading run of the span is always exactly $a_1^{m_1}$ and the trailing
run is exactly $a_r^{m_r}$.

\emph{Base cases.}
\begin{itemize}
  \item $r = 1$: the span is one run, so \textsc{Peel} emits a single
        terminal token $\tau_{k-1} = (a_1^{m_1},\, 0)$.  Count: $h = 1$.
  \item $r = 2$: the leading and trailing runs are adjacent and cover the
        entire span, so \textsc{Peel} emits the terminal token
        $\tau_{k-1} = (a_1^{m_1} a_2^{m_2},\, 0)$.  Count: $h = 1$.
\end{itemize}

\emph{Inductive step} ($r \ge 3$).  \textsc{Peel} emits one content token
$\tau_1 = (a_1^{m_1} a_r^{m_r},\, m_1)$ and recurses on the middle, which
is the RLE $a_2^{m_2} \cdots a_{r-1}^{m_{r-1}}$ with $r - 2$ runs and the
alternation property preserved.  By the inductive hypothesis the
recursion contributes $\lceil (r-2)/2 \rceil$ tokens, so the total is
\[
  1 \;+\; \lceil (r-2)/2 \rceil \;=\; \lceil r/2 \rceil \;=\; h.
\]
Each inductive level pairs the current outermost surviving runs, so the
$j$-th content token records runs $j$ and $r+1-j$ exactly as claimed.

\emph{Bound.}  Since every run has length $\ge 1$, $r \le n$, giving
$k \le 1 + \lceil n/2 \rceil$; equality holds iff every run has length
exactly~$1$.
\end{proof}

\Cref{fig:pairing} shows the pairing for \texttt{CASSAYFF}: runs~$1$
and~$6$ form the outer arc and become token~$\tau_1$, runs~$2$ and~$5$
form the inner arc and become $\tau_2$, and the two middle runs are
unpaired, together forming the kernel~$\tau_3$.  The picture also makes
\Cref{fig:example} easier to read: the peeling at successive depths
reads off the arcs from outside in.

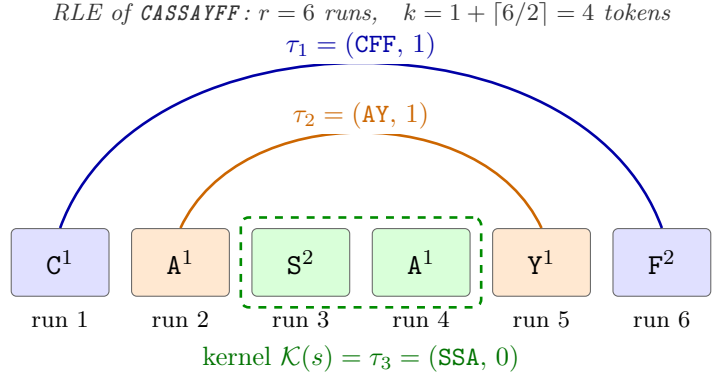
\begin{figure}[ht]
\centering
\begin{tikzpicture}[
  run/.style={
    draw=black!60,
    rounded corners=2pt,
    minimum width=1.3cm,
    minimum height=0.9cm,
    font=\ttfamily,
    line width=0.4pt
  },
  arc/.style={draw, line width=1pt},
  lbl/.style={font=\small, fill=white, inner sep=2pt}
]
\node[run, fill=blue!12]    (r1) at (0,   0) {$\mathtt{C}^{1}$};
\node[run, fill=orange!18]  (r2) at (1.6, 0) {$\mathtt{A}^{1}$};
\node[run, fill=green!15]   (r3) at (3.2, 0) {$\mathtt{S}^{2}$};
\node[run, fill=green!15]   (r4) at (4.8, 0) {$\mathtt{A}^{1}$};
\node[run, fill=orange!18]  (r5) at (6.4, 0) {$\mathtt{Y}^{1}$};
\node[run, fill=blue!12]    (r6) at (8.0, 0) {$\mathtt{F}^{2}$};

\foreach \i/\x in {1/0, 2/1.6, 3/3.2, 4/4.8, 5/6.4, 6/8.0}
  \node[font=\footnotesize, below=2pt] at (\x, -0.45) {run $\i$};

\draw[arc, blue!65!black]
  (r1.north) to[out=70, in=110]
  node[lbl, pos=0.5, above=-1pt] {$\tau_1 = (\mathtt{CFF},\, 1)$}
  (r6.north);

\draw[arc, orange!80!black]
  (r2.north) to[out=65, in=115]
  node[lbl, pos=0.5, above=-1pt] {$\tau_2 = (\mathtt{AY},\, 1)$}
  (r5.north);

\begin{scope}[on background layer]
  \node[fit=(r3)(r4), inner sep=4pt,
        draw=green!55!black, line width=1pt,
        dashed, rounded corners=3pt] (K) {};
\end{scope}
\node[below=0.35cm of K.south, font=\small, green!45!black]
  {kernel $\mathcal{K}(s) = \tau_3 = (\mathtt{SSA},\, 0)$};

\node[font=\footnotesize\itshape, black!80] at (4.0, 3.3)
  {RLE of \texttt{CASSAYFF}\,:\ $r=6$ runs,\quad $k = 1 + \lceil 6/2 \rceil = 4$ tokens};
\end{tikzpicture}
\caption{Run-pairing view of \texttt{CASSAYFF}.  Each content token
corresponds to one arc pairing two runs of the RLE (run~$j$ with
run~$r+1-j$); the middle run or two middle runs form the peeling kernel.
The token count $k = 1 + \lceil r/2 \rceil$ of \Cref{thm:run_pairing}
counts the sentinel plus the arcs.}
\label{fig:pairing}
\end{figure}

\paragraph{Immediate consequences.}
Three natural string transformations act predictably on the token
sequence, as one-line corollaries of \Cref{thm:run_pairing}.
\emph{Reversal:} $\mathcal{F}(s^R)$ has the same number of tokens as
$\mathcal{F}(s)$, with each content token's front and back swapped and
the terminal replaced by the reversed kernel $\mathcal{K}(s)^R$.
\emph{Uniform run-length dilation} by the morphism $h_c(a) = a^c$
($c \ge 1$): every $m_i$ scales by $c$, so every split position scales
by $c$ and every symbol string $\sigma_j$ is replaced by
$h_c(\sigma_j)$.  \emph{Tandem repetition} $s = w^m$ with $r(w) \ge 2$
and the first and last runs of $w$ distinct: no boundary runs merge,
and the non-terminal tokens are periodic with period $r(w)$.

\subsection{Optimality of Maximal Peeling}

Flashback always peels the entire boundary run on each side.  A natural
question is whether peeling less aggressively (taking only part of the
front or back run at some level) could reduce the token count.  It cannot:
the greedy choice is in fact globally optimal.  To make the question precise
we first say what "admissible" means.

\begin{definition}[Admissible bilateral run-peeling]
\label{def:admissible}
An \emph{admissible bilateral run-peeling decomposition} of a non-empty
string $s$ is built by applying the following rule to an active span until
the span is consumed.
\begin{itemize}
  \item \emph{Termination.}  If the active span consists of a single run or
        two adjacent runs, emit one terminal token for the span and stop.
  \item \emph{Peeling step.}  Otherwise let $L$ be the length of the span's
        leading run and $R$ the length of its trailing run.  Choose any
        integers $x \in \{1, \ldots, L\}$ and $y \in \{1, \ldots, R\}$,
        emit the token formed by the peeled prefix of length $x$ and the
        peeled suffix of length $y$, and recurse on the remaining middle.
\end{itemize}
Flashback is the special case $x = L$, $y = R$ at every step: the two
boundary runs are always consumed in full.
\end{definition}

\begin{theorem}[Greedy optimality]
\label{thm:optimality}
Every admissible bilateral run-peeling decomposition of $s$ uses at least
$\lceil r/2 \rceil$ content tokens, and Flashback attains this minimum.
\end{theorem}

\begin{proof}
The idea is to track how many of the original RLE runs of $s$ the active
span still intersects as peeling proceeds.  Let $N$ denote that count;
initially $N = r$.

\emph{Step 1: each non-terminal step decreases $N$ by at most~$2$.}
A peeling step removes a prefix of length $x \ge 1$ from the current
leading run and a suffix of length $y \ge 1$ from the current trailing run.
Advancing the left endpoint of the span by $x$ positions moves it across
at most one original run boundary, so at most one original run falls out
of the span on the left; symmetrically on the right.

\emph{Step 2: at termination, $N \le 2$.}  The termination rule fires
only when the active span is one run or two adjacent runs, hence
intersects at most two original RLE runs.

\emph{Step 3: arithmetic.}  Suppose the decomposition emits $t$ content
tokens in total.  The last is the terminal token, so there are exactly
$t - 1$ non-terminal steps.  By Steps~1 and~2,
\[
  r \;-\; 2(t - 1) \;\le\; N_{\text{final}} \;\le\; 2,
\]
which rearranges to $t \ge 1 + (r - 2)/2$.  Since $t$ is an integer this
gives $t \ge \lceil r/2 \rceil$.

By \Cref{thm:run_pairing}, Flashback produces exactly $\lceil r/2 \rceil$
content tokens, so the lower bound is tight.
\end{proof}

\begin{corollary}[Even-run uniqueness]
\label{cor:optimality_even}
If $r$ is even, Flashback is the unique token-minimal admissible
decomposition.
\end{corollary}

\begin{proof}
Write $r = 2m$.  A token-minimal decomposition uses $m$ content tokens
and hence $m - 1$ non-terminal steps.  The inequality $r - 2(t - 1) \le 2$
from the proof of \Cref{thm:optimality} is then an equality, so every
non-terminal step must reduce $N$ by exactly $2$.  This forces each step
to remove an original run from each side, i.e.\ to peel both boundary
runs completely, which is the Flashback rule.
\end{proof}

\subsection{The Peeling Kernel}

After all the pairing is done, one small piece is left over: the terminal
token's symbol string.  This residue is what remains once bilateral peeling
can no longer make progress, and we call it the \emph{peeling kernel}.

\begin{definition}[Peeling kernel]
\label{def:kernel}
The \emph{peeling kernel} of $s$ is $\mathcal{K}(s) = \sigma_{k-1}$, the
terminal token's symbol string.
\end{definition}

\begin{corollary}[Kernel structure]
\label{thm:kernel_alphabet}
With $h = \lceil r/2 \rceil$ as in \Cref{thm:run_pairing}, the peeling
kernel is
\[
  \mathcal{K}(s) \;=\;
  \begin{cases}
    a_h^{m_h} & \text{if } r \text{ is odd,} \\[3pt]
    a_h^{m_h} \cdot a_{h+1}^{m_{h+1}} & \text{if } r \text{ is even.}
  \end{cases}
\]
In either case $|\mathrm{alphabet}(\mathcal{K})| \le 2$.
\end{corollary}

\begin{proof}
The formula for $\mathcal{K}(s)$ is the symbol string of the terminal
token $\tau_{k-1}$ in \Cref{thm:run_pairing}.  Each case involves at
most two distinct symbols.
\end{proof}

\paragraph{Palindromes.}
Flashback recovers a classical fact about run-length encoding: a string
$s$ is a palindrome if and only if $r$ is odd and its RLE is symmetric
($a_j = a_{r+1-j}$ and $m_j = m_{r+1-j}$ for $1 \le j \le \lfloor r/2 \rfloor$).
The equivalence follows from RLE uniqueness alone and does not use the
run-pairing theorem; we mention it because in Flashback terms it has a
clean form: $\mathcal{K}(s)$ is a single run and every non-terminal
token has matching front and back symbols with equal-length pairs.

\subsection{Edit Locality}

Because $k$ is a function of $r$ alone and each content token depends
only on its paired pair of runs, edits have sharply local effects.  The
cleanest statement concerns edits that preserve the RLE skeleton,
i.e.\ changes to run lengths that do not alter run symbols.

\begin{proposition}[Run-length locality]
\label{prop:run_length_locality}
Let $s = a_1^{m_1} \cdots a_r^{m_r}$ and $t = a_1^{m'_1} \cdots a_r^{m'_r}$
have the same run symbols $a_1, \ldots, a_r$ (possibly with different
multiplicities), and let
\[
  D \;=\; \{\, i : m_i \ne m'_i \,\}.
\]
Then $\mathcal{F}(s)$ and $\mathcal{F}(t)$ differ at exactly the
content-token depths
\[
  \{\, \min(i,\, r + 1 - i) : i \in D \,\}.
\]
In particular, changing a single run length $m_i$ rewrites exactly one
content token, at depth $\min(i,\, r + 1 - i)$.
\label{cor:single_run_locality}
\end{proposition}

\begin{proof}
By \Cref{thm:run_pairing}, the content token at depth $d$ pairs runs $d$
and $r + 1 - d$ (for non-terminal depths), and the terminal token
records the middle run or two middle runs.  In both cases the token is
determined entirely by the lengths of its paired runs.  Consequently,
$\tau_d \ne \tau'_d$ if and only if at least one paired run has a
different length in $s$ and $t$.

It remains to identify the depth at which run $i$ is paired.  Run $i$
shares a token with run $r + 1 - i$, and that token sits at the smaller
of the two indices, namely $\min(i,\, r + 1 - i)$.  Hence $\tau_d \ne
\tau'_d$ iff $d = \min(i,\, r + 1 - i)$ for some $i \in D$, giving
exactly the displayed set.
\end{proof}

Character-level edits are weaker but still bounded.  A single insertion,
deletion, or substitution changes at most two run boundaries (merging or
creating them), so $|\Delta r|\le 2$ and, since $k=1+\lceil r/2\rceil$,
$|\Delta k|\le 1$.  The token \emph{count} is therefore stable under
local character edits, though individual token \emph{contents} may shift
at several depths because the deeper active spans realign.

\subsection{Image Characterisation}

We close by identifying which token sequences actually arise.  The
constraints are all local: each content token is a pair of runs, the
terminal is one or two runs, and consecutive depths must alternate
symbols on each side.  These conditions can be checked in a single pass,
so membership in $\mathrm{Im}(\mathcal{F})$ is decidable in linear time.

\begin{theorem}[Valid token sequences]
\label{thm:image}
Let $T = [(\texttt{@\$}, 1),\, (\sigma_1, p_1),\, \ldots,\, (\sigma_{k-1}, p_{k-1})]$
be a candidate token sequence with $k \ge 2$ and all non-sentinel
characters drawn from $\Sigma$.  For each non-terminal depth
$1 \le d \le k-2$, write
\[
  a_d \;=\; \text{front symbol of $\tau_d$} \;=\; \sigma_d[1],
  \qquad
  c_d \;=\; \text{back symbol of $\tau_d$} \;=\; \sigma_d[|\sigma_d|].
\]
Then $T$ lies in the image of $\mathcal{F}$ if and only if the following
conditions hold.
\begin{enumerate}
  \item \emph{Token form.}  For each non-terminal depth $1 \le d \le k-2$:
        \begin{itemize}
          \item $p_d \ge 1$ and $|\sigma_d| - p_d \ge 1$ (both the front
                part $\sigma_d[1..p_d]$ and the back part
                $\sigma_d[p_d+1..|\sigma_d|]$ are non-empty);
          \item the front part is a single run of $a_d$, and the back
                part is a single run of $c_d$.
        \end{itemize}
  \item \emph{Terminal form.}  $p_{k-1} = 0$ and $\sigma_{k-1}$ is either
        a single run $a^m$ or two adjacent runs $a^\ell c^b$ with $a \ne c$.
  \item \emph{Alternation} (only when $k \ge 3$):
        \begin{itemize}
          \item \emph{Interior:} for $1 \le d \le k-3$,
                $a_{d+1} \ne a_d$ and $c_{d+1} \ne c_d$.
          \item \emph{Boundary:} the first character of $\sigma_{k-1}$
                differs from $a_{k-2}$, and the last character of
                $\sigma_{k-1}$ differs from $c_{k-2}$.
        \end{itemize}
\end{enumerate}
When $k = 2$, only conditions (1) and (2) apply.  Moreover,
$\mathcal{F}(\mathcal{F}^{-1}(T)) = T$ for every valid $T$.
\end{theorem}

\begin{proof}[Proof sketch]
\emph{Necessity} is immediate from \Cref{thm:run_pairing} (giving (1)
and (2)) and RLE alternation (giving (3)).  \emph{Sufficiency} uses a
direct reassembly: the fronts, terminal, and reversed backs concatenate
into a string whose RLE is forced to be valid by conditions (1)--(3);
the full reassembly is given in \Cref{app:image_proof}.
\end{proof}

\section{Related Work}
\label{sec:related_work}

\paragraph{Run-length encoding.}
RLE replaces each maximal run $a^k$ with the pair $(a, k)$.
\Cref{thm:run_pairing} shows that Flashback is, in a precise sense, a
bilateral pairing of the RLE: it pairs the first run with the last, the second
with the second-to-last, and so on.  Where RLE processes runs left-to-right,
Flashback processes them outside-in, capturing the nesting structure that RLE
discards.  The token count $k = 1 + \lceil r/2 \rceil$ is determined entirely
by the number of runs~$r$.

\paragraph{Lempel--Ziv decompositions.}
The LZ76 factorisation~\cite{lz76} and its descendants
LZ77~\cite{lz77}/LZ78~\cite{lz78} parse left-to-right, greedily extending
each factor against a dictionary of previously seen substrings.  LZ exploits
\emph{repetitive} structure across positions; Flashback exploits
\emph{bilateral boundary} structure.  The two are complementary: LZ76 produces
$O(n / \log n)$ factors for random strings~\cite{lz76}, while Flashback
produces $\Theta(n)$ tokens.  Flashback is not a compression scheme.

\paragraph{Burrows--Wheeler Transform.}
The BWT~\cite{bwt} is a reversible permutation that clusters identical
characters.  It does not produce variable-length tokens.  Flashback produces
an ordered token sequence with split positions that directly encode
reconstruction.

\paragraph{Bidirectional macro schemes.}
Storer and Szymanski~\cite{storer1982} introduced bidirectional macro schemes,
where substrings may reference each other in both directions.  Computing the
smallest such scheme is NP-hard.  Flashback uses no references: each token is
self-contained.

\paragraph{Lyndon factorisation.}
The Lyndon factorisation~\cite{lyndon1954,duval1983} decomposes a string into
a non-increasing sequence of Lyndon words in linear time.  It is based on
lexicographic properties and produces a sequential factorisation; Flashback's
tokens encode bilateral boundary structure with an explicit nesting.

\section{Conclusion}
\label{sec:conclusion}

\textsc{Flashback} is a small, deterministic decomposition that runs in
$O(n)$ time and space and is exactly reversible.  Once phrased at the run
level, it turns out to be bilateral pairing of the run-length encoding
(\Cref{thm:run_pairing}), and that single observation is what carries the
rest of the paper: the $1+\lceil r/2\rceil$ token count, the matching
lower bound over all admissible bilateral run-peeling schemes
(\Cref{thm:optimality}), the two-symbol kernel, the palindrome test,
and the image characterisation all follow as corollaries.  One less
expected property is token-level edit locality: under a
skeleton-preserving run-length change, exactly one content token is
rewritten (\Cref{prop:run_length_locality}).  We see Flashback primarily
as a structural tool for reasoning about run-level symmetry rather than
as a competitor to existing string-compression or indexing schemes.

\bibliographystyle{plain}
\bibliography{references}

\appendix
\section{Supplementary Details}
\label{app:details}

\subsection{Iterative Formulation of Decomposition}
\label{app:iterative_decompose}

Since \textsc{Peel} is tail-recursive (each invocation makes at most one
recursive call, as the last operation), it admits a direct iterative rewrite:

\begin{algorithm}[ht]
\caption{\textsc{Flashback-Decompose-Iterative}$(s)$}
\label{alg:decompose_iter}
\begin{algorithmic}[1]
\Require String $s \in \Sigma^*$
\Ensure Ordered token sequence $T$
\State $\hat{s} \gets \texttt{@} \cdot s \cdot \texttt{\$}$
\State $T \gets [\,]$;\quad $\mathit{lo} \gets 1$;\quad
       $\mathit{hi} \gets |\hat{s}|{+}1$
\While{$\mathit{lo} < \mathit{hi}$}
    \State $\ell \gets$ leading run length at $\hat{s}[\mathit{lo}]$
    \If{$\ell = \mathit{hi} - \mathit{lo}$}
        \State Append $(\hat{s}[\mathit{lo}..\mathit{hi}),\; 0)$ to $T$
        \State \textbf{break}
    \EndIf
    \State $r \gets$ trailing run start for $\hat{s}[\mathit{hi}-1]$
    \State $\sigma \gets \hat{s}[\mathit{lo}..\mathit{lo}+\ell) \;\|\;
                         \hat{s}[r..\mathit{hi})$
    \State $\mathit{mid\_lo} \gets \mathit{lo} + \ell$;\quad
           $\mathit{mid\_hi} \gets r$
    \If{$\mathit{mid\_lo} \ge \mathit{mid\_hi}$}
        \State Append $(\sigma,\; 0)$ to $T$
        \State \textbf{break}
    \Else
        \State Append $(\sigma,\; \ell)$ to $T$
    \EndIf
    \State $\mathit{lo} \gets \mathit{mid\_lo}$;\quad
           $\mathit{hi} \gets \mathit{mid\_hi}$
\EndWhile
\State \Return $T$
\end{algorithmic}
\end{algorithm}

This version uses $O(1)$ stack space (beyond the output list).

\subsection{Random-String Statistics}
\label{app:random_string_stats}

With notation as in \Cref{thm:expected_k}: $n\ge 2$, alphabet size
$\sigma\ge 2$, $p=(\sigma-1)/\sigma$, $m=n-1$, and $q=1-2p=(2-\sigma)/\sigma$.

\begin{corollary}[Variance of $k$]
\label{cor:var_k}
\[
\mathrm{Var}[k] \;=\; \frac{m p(1-p)}{4}
\;-\; \frac{m p(1-p)\,q^{m-1}}{2}
\;+\; \frac{1-q^{2m}}{16}.
\]
The cross-term accounts for the covariance between $X$ and
$\mathbf{1}_{X\text{ even}}$.  In particular, $\mathrm{Var}[k]=0$ when $n=2$
(every length-$2$ string has $k=2$); for large~$n$,
$\mathrm{Var}[k]\approx(n-1)(\sigma-1)/(4\sigma^2)$.
\end{corollary}

\begin{corollary}[Kernel alphabet statistics]
\label{cor:kernel_random}
$|\mathrm{alphabet}(\mathcal{K})|=1$ iff $r$ is odd iff $X$ is even, and
$|\mathrm{alphabet}(\mathcal{K})|=2$ otherwise.  Hence
\[
\Pr\bigl(|\mathrm{alphabet}(\mathcal{K})|=1\bigr)
\;=\; \tfrac{1}{2}\bigl(1+q^{n-1}\bigr),
\]
which equals $1/2$ exactly for $\sigma=2$, $n\ge 2$, and approaches $1/2$
exponentially for $\sigma\ge 3$.
\end{corollary}

\subsection{Full Proof of Image Characterisation (\texorpdfstring{\Cref{thm:image}}{Theorem on valid token sequences})}
\label{app:image_proof}

\begin{proof}[Proof of \Cref{thm:image}]
\emph{Necessity.}  Suppose $T = \mathcal{F}(s)$ for some
$s = a_1^{m_1} \cdots a_r^{m_r}$.  By \Cref{thm:run_pairing}, each
non-terminal token $\tau_d$ concatenates run $d$ with run $r+1-d$, so
its front is a run of $a_d$ and its back is a run of $a_{r+1-d}$; this
gives~(1) with $c_d = a_{r+1-d}$.  The terminal records the middle run
(odd $r$) or the two middle runs (even $r$), giving~(2).  For~(3), the
interior alternation $a_{d+1} \ne a_d$ and $c_{d+1} \ne c_d$ is just
the RLE alternation $a_i \ne a_{i+1}$ applied to the paired positions,
and the boundary condition is the same alternation between the last
non-terminal pair and the middle run(s).

\emph{Sufficiency.}  Given $T$ satisfying (1)--(3), we construct a
string $s$ with $\mathcal{F}(s) = T$.  Read off one front run
$a_d^{p_d}$ and one back run $c_d^{|\sigma_d| - p_d}$ from each
non-terminal $\tau_d$, and read the middle run (or two middle runs)
from the terminal $\tau_{k-1}$.  Arrange them in a single string by
placing the fronts on the left in depth order, the terminal in the
middle, and the backs on the right in reverse depth order:
\[
  s \;=\; a_1^{p_1}\, a_2^{p_2} \cdots a_{k-2}^{p_{k-2}}
          \;\|\; \sigma_{k-1} \;\|\;
          c_{k-2}^{|\sigma_{k-2}| - p_{k-2}} \cdots c_1^{|\sigma_1| - p_1}.
\]
By the interior alternation, consecutive fronts and consecutive backs
have distinct symbols; by the boundary alternation, the terminal's
first and last characters differ from the adjacent $a_{k-2}$ and
$c_{k-2}$; and by the terminal form~(2), its two runs (if any) already
differ.  Hence $s$ has a valid RLE, and \Cref{thm:run_pairing} applied
to $s$ reproduces $T$ exactly.
\end{proof}

\end{document}